\documentclass[11pt]{article}

\usepackage[hmargin=1in,vmargin=1in]{geometry}
\usepackage{amsmath,amsthm,amssymb,amsfonts,verbatim,mathptmx}

\usepackage{fullpage}
\usepackage{graphicx}

\usepackage[dvipsnames,usenames]{color}
\usepackage{hyperref}
\usepackage{todonotes}
\DeclareMathAlphabet{\mathcal}{OMS}{cmsy}{b}{n}
\DeclareMathAlphabet{\mathcal}{OMS}{cmsy}{m}{n}

 \providecommand{\F}{\mathbb{F}}

\newtheorem{lemma}{Lemma}[section]
\newtheorem{theorem}[lemma]{Theorem}
\newtheorem{cor}[lemma]{Corollary}

\newtheorem{defn}{Definition}

\theoremstyle{remark}

\renewcommand{\epsilon}{\varepsilon}
\renewcommand{\le}{\leqslant}
\renewcommand{\ge}{\geqslant}

\def\RR{\mathbb{R}}
\def\ZZ{\mathbb{Z}}

\def\FF{\mathbb{F}}

\def \mA {\mathcal{A}}

\def \mC {\mathcal{C}}

\def \mF {\mathcal{F}}

\def \mR {\mathcal{R}}
\def \mS {\mathcal{S}}

\newcommand{\Gb}{\beta}
     
\newcommand{\Gd}{\delta}     
\newcommand{\Ge}{\epsilon}

\def \bc {{\bf c}}

\def \bx {{\bf x}}

\def \by {{\bf y}}

\def \bu {{\bf u}}
\def \bv {{\bf v}}

\def\rank{\rm rank}

\date{}

\title{Asymptotic Gilbert-Varshamov bound on Frequency Hopping Sequences}

\author{Xianhua Niu\thanks{School of Computer and Software Engineering, Xihua University and National Key Laboratory of Science and Technology on Communications, University of Electronic Science and Technology of China. Research supported in part by the National Science Foundation
of China(Grant No. 61401369), the Youth Science and Technology
Fund of Sichuan Province(No.2017JQ0059). Email: {\tt rurustef1212@gmail.com.}} \and Chaoping Xing\thanks{School of Physical and Mathematical Sciences, Nanyang Technological University, Singapore. Email: {\tt xingcp@ntu.edu.sg}.} \and Chen Yuan\thanks{Centrum Wiskunde \& Informatica, Amsterdam, Netherlands. Email: {\tt Chen.Yuan@cwi.nl}} }

\begin{document}

\maketitle


%
%
%

\begin{abstract}
Given a $q$-ary frequency hopping sequence set of length $n$ and size $M$ with Hamming correlation $H$, one can obtain a $q$-ary (nonlinear) cyclic code of length $n$ and size $nM$ with Hamming distance $n-H$.  Thus, every upper bound on the size of a code from coding theory gives an upper bound on the size of a frequency hopping sequence set. Indeed,  all upper bounds from coding theory have been converted to upper bounds on frequency hopping sequence sets (\cite{Ding09}). On the other hand, a lower bound from  coding theory does not automatically produce a lower bound for frequency hopping sequence sets. In particular, the most important lower bound--the Gilbert-Varshamov bound in coding theory has not been transformed to frequency hopping sequence sets. The purpose of this paper is to convert the Gilbert-Varshamov bound in coding theory to frequency hopping sequence sets by establishing a connection between a special family of cyclic codes (which are called hopping cyclic codes in this paper) and frequency hopping sequence sets.
We provide two proofs of the Gilbert-Varshamov bound. One is based on probabilistic method that requires advanced tool--martingale. This proof covers the whole rate region.
The other  proof is purely elementary but only covers part of the rate region.
\end{abstract}

\section{Introduction}
Frequency hopping (FH) sequences are designed for transmitting radio signals between transmitter and receiver.
To mitigate the interference caused by hits of frequencies, some systems require the set of FH sequences with low Hamming correlation and large size and others require a single FH sequence with good Hamming autocorrelation~\cite{Fan,Golomb,Simon}. To judge the performance of an FH sequence set, we need to explore its theoretical limit in terms of five parameters, the size of available frequencies (alphabet size), sequence length, family size, maximum Hamming autocorrelation and maximum Hamming crosscorrelation.


Lempel and Greenberger ~\cite{LG} established a lower bound on the maximum Hamming autocorrelation of an FH sequence. Peng and Fan~\cite{PF} obtained a lower bound on the maximum Hamming correlation of an FH sequence set (or an upper bound on family size). In fact, this bound can be obtained from the Plotkin bound in coding theory. As an FH sequence set gives a nonlinear cyclic code, every upper bound on the size of a code in coding theory  can be converted into an upper bound on the size of an FH sequence set. Indeed, Ding et.al ~\cite{Ding09} presented three upper bounds on the family size of the FH sequence set derived from the classical bounds in coding theory. By reformulating the Singleton bound in~\cite{Ding09}, Yang et.al \cite{Yang} obtained a new and tighter lower bound on the maximum Hamming correlation of an FH sequence set. Liu et.al \cite{Liu} further improved this bound.
Note that a lower bound on maximum Hamming correlation also implies an upper bound on the family size.
There are many attempts to construct optimal sets of FH sequences meeting these upper bounds on the family size \cite{PF,Ding09,Ding10,Yang,Golomb}.

 On the other hand, a lower bound from  coding theory does not automatically produce a lower bound for frequency hopping sequence sets. In particular, the most important lower bound--the Gilbert-Varshamov bound in coding theory has not been transformed to frequency hopping sequence sets. The Gilbert-Varshamov bound in coding theory is asymptotically good in the sense that the length $n$ tends to $\infty$ and the  alphabet size $q$ is fixed. This scenario is interesting from both practical and theoretical points of view for FH sequence sets. In this scenario, most of the upper bounds on  FH sequence sets mentioned above are not achievable.
The Gilbert-Varshamov bound in coding theory characterizes the existence of good codes with positive relative distance and positive rate over a fixed alphabet size. Considering the strong connection between FH sequences and cyclic codes, one might be curious about whether there exists a similar lower bound for the FH sequence set. We emphasize that the classic Gilbert-Varshamov bound does not hold for cyclic codes (in fact it is big open problem in coding theory about whether there are linear cyclic codes with positive relative distance and positive rate over a fixed alphabet size).
 In this paper, we present a class of cyclic codes called hopping cyclic codes which can be converted into sets of FH sequences. We show that there exists a class of hopping cyclic codes attaining Gilbert-Varshamov bound. This also implies that there exists a class of FH sequence sets attaining the same bound. In this sense, our results characterize the asymptotic behaviour of an FH sequence set.

The rest of our paper is organized as follows. In Section $2$, we introduce some notations and known bounds on an FH sequence set.
In Section $3$, we obtain the asymptotic version of upper bound from known upper bounds on an FH sequence set. In Section $4$,
we prove the Gilbert-Varshamov bound on an FH sequence set by showing the existence of hopping cyclic codes attaining the Gilbert-Varshamov bound.


\section{Preliminaries}

\subsection{Notations and upper bounds on  FH sequences}

Let $F= \{f_{1}, f_{2},\ldots, f_{q}\}$ be a frequency slot
set which is also regarded as an alphabet set. Let $\mS$ be the collection of all sequences of length $n$ over $F$.
We call the element in $\mS$ an FH sequence.
Given two FH sequences $\bx$=($x_{0}$, $x_{1}$,\ldots, $x_{n-1}$), $\by$=
($y_{0}$, $y_{1}$,\ldots, $y_{n-1}$), the Hamming
correlation function $H_{\bx,\by}(t)$ with time delay $\tau$ is defined to be
\begin{eqnarray}
H_{\bx,\by}(\tau)=\sum^{n-1}_{i=0}h(x_{i},y_{i+\tau}), \quad 0\leq \tau<n,
\end{eqnarray}
where $h(a, b)=1$ if $a=b$, and $h(a, b)=0$ otherwise. All the operations among
the indices are performed by $\bmod n$.

Let $\mF$ be a subset of $\mS$.
Define the maximum Hamming autocorrelation $H_{a}(\mF)$, the maximum Hamming cross correlation $H_{c}(\mF)$, the maximum Hamming
correlation $H_{m}(\mF)$ of set $\mF$ as follows:
\begin{eqnarray*}
&&H_{a}(\mF)=\max_{0<\tau<n}\{H_{\bx,\bx}(\tau)\},\quad H_{c}(\mF)=\max_{0<\tau\leq n, \bx\neq\by\in \mF}\{H_{\bx,\by}(\tau)\}\\
&&H_{m}(\mF)=\max\{H_a(\mF), H_c(\mF)\}.
\end{eqnarray*}

In this paper, we use $(n,M,\lambda;q)$ to denote a set $\mF$ that consists of $M$ frequency hopping
sequences of length $n$ over alphabet size $q$ has the maximum Hamming
correlation $H_{m}(\mF)=\lambda$.

\begin{defn}
Let $\mF$ be an FH sequence set with parameters $(n,M,\lambda;q)$. The information rate $r(\mF)$ and the relative Hamming correlation $\delta_H(\mF)$ is
$$
r(\mF)=\frac{\log_qM}{n}, \quad \delta_H(\mF)=\frac{\lambda}{n}.
$$
\end{defn}

To capture the asymptotic behaviour of code,  it is natural to define the information rate and its relative Hamming correlation of an FH sequence set as a corresponding version of code rate and relative distance.\footnote{Precisely speaking, the relative Hamming correlation is $1-$``relative distance''.}

\begin{defn}
For fixed integer $q\ge 2$, the fundamental domain of  frequency hopping  sequences is defined to be
\[\begin{split}
D_q=\{(\Gd_H,r)\in\RR^2:\; \mbox{there exists a family   $\F=\{\mF_i\}_{i=1}^{\infty}$  of FH sequence sets such that} \\ \mbox{the length $n_i$ of $\mF_i$ tends to $\infty$,  $\lim_{i\rightarrow\infty}\frac{\lambda(\mF_{i})}{n_i}=\Gd_H$ and $\lim_{i\rightarrow\infty}\frac{\log_q|\mF_i|}{n_i}=r$.}\}\end{split}\]
\end{defn}
Note that for a given family   $\F=\{\mF_i\}_{i=1}^{\infty}$  of FH sequence sets, the limits $\lim_{i\rightarrow\infty}\frac{\lambda(\mF_i)}{n_i}$ and $\lim_{i\rightarrow\infty}\frac{\log_q|\mF_i|}{n_i}$ may not exist. However, one can always find a subfamily $\mathbb{E}=\{\mF_{i_j}\}_{j=1}^{\infty}$ of $\F$ such that both $\lim_{j\rightarrow\infty}\frac{\lambda(\mF_{i_j})}{n_{i_j}}$ and $\lim_{j\rightarrow\infty}\frac{\log_q|\mF_{i_j}|}{n_{i_j}}$ exist. Furthermore, we have
$\lim_{j\rightarrow\infty}\frac{\lambda(\mF_{i_j})}{n_{i_j}}\ge\liminf_{i\rightarrow\infty}\frac{\lambda(\mF_{i})}{n_i}$ and $\lim_{j\rightarrow\infty}\frac{\log_q|\mF_{i_j}|}{n_{i_j}}\ge\liminf_{i\rightarrow\infty}\frac{\log_q|\mF_i|}{n_i}$.

As in coding theory, it is straightened to see that there is function $\Gb_q(\Gd_H)$ from $[0,1]$ to $[0,1]$ such that
\[D_q=\{(\Gd,r)\in[0,1]\times[0,1]:\; r\le \Gb_q(\Gd_H)\}.\]
By the Plotkin bound \cite{Ding09} or the Peng-Fan bound \cite{PF}, one immediately has $\Gb_q(\Gd_H)=0$ for all $0\le\Gd_H\le\frac1q$. As in coding theory, it is an important open problem to determine the function $\Gb_q(\Gd_H)$. However, as one can imagine, determining the function $\Gb_q(\Gd_H)$ is a challenging and difficult problem. Instead, one would be interested in lower our upper bounds on the function $\Gb_q(\Gd_H)$. Various  upper bounds on the function $\Gb_q(\Gd_H)$ have been obtained from coding theory. The various upper bounds of finite version given in \cite{Ding09} can be easily converted into their asymptotic versions.

\begin{theorem}[The asymptotic upper bounds on FH sequences \cite{Ding09}]
For $q\ge 2$, one has
\begin{itemize}
\item[{\rm (i)}] The Singleton bound
\[\Gb_q(\Gd_H)\le \Gd_H.\]
\item[{\rm (ii)}]
The  Plotkin bound on FH sequences:
\begin{align*}
&\Gb_q(\Gd_H)=0, \quad  \mbox{for}\ 0\leq \delta_H \leq \frac{1}{q},\\
&\Gb_q(\Gd_H)\leq -\frac{1}{q-1}+\frac{q}{q-1}\delta_H, \quad \mbox{for}\  \frac{1}{q}\leq \delta_H \leq 1.
\end{align*}
\item[{\rm (iii)}] The Sphere-packing bound
\begin{equation*}
\Gb_q(\Gd_H)\leq 1-H_q\left(\frac{1-\delta_H}{2}\right)
\end{equation*}
where $H_q(x)=x\log_q (q-1)-x\log_q x-(1-x)\log_q(1-x)$.
\item[{\rm (iv)}] The Linear Programming bound 
\begin{equation*}
\Gb_q(\Gd_H)\leq H_q\left(\frac{q-1-(q-2)(1-\delta_H)-2\sqrt{\delta_H(\mF)(1-\delta_H)(q-1)}}{q}\right).
\end{equation*}
\end{itemize}
\end{theorem}

Figure 1 shows the relation between the function $\Gb_q(\Gd_H)$ and various upper bounds.

\begin{figure}[h]
\centering
\includegraphics[scale=0.56]{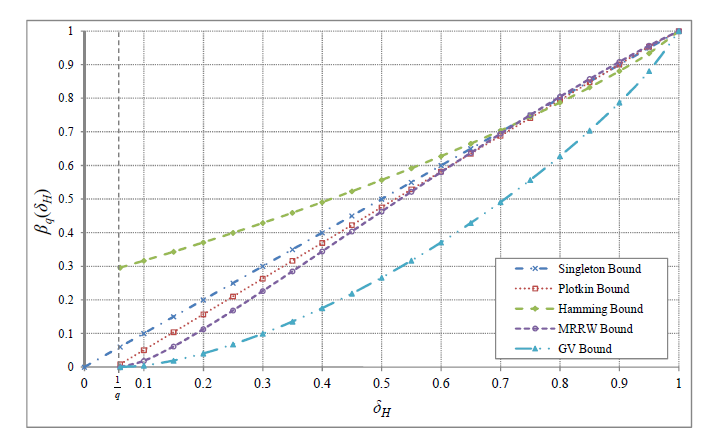}
\caption{Asymptotic upper and lower bounds on FH sequences for $q=17$.}
\centering
\end{figure}

As far as we know, no lower bounds on the function $\Gb_q(\Gd_H)$ are known in literature. The main purpose of this paper is to convert the Gilbert-Varshamov bound from coding theory  into a lower bounds for FH sequence sets.

\subsection{Connection between FH sequences and codes}

Let $\bu=(u_0, u_1,\cdots, u_{n-1})\in \ZZ_q^n$. Define $\bu_i=(u_{i}, u_{i+1}, \cdots, u_{i-1})$ to be the $i$-th cycle shift of vector $\bu$.
The Hamming distance between $\bu$ and $\bu_i$ is denoted by $d(\bu,\bu_i)$.
Define $d(\bu)=\min_{0<i<n}d(\bu,\bu_i)$.
$\bu$ can also be treated as an FH sequence of length $n$ over a alphabet size $q$.
Then, the Hamming autocorrelation $H_{\bu,\bu}(i)$ of $\bu$ at time delay $i$ becomes $n-d(\bu,\bu_i)$.
It follows that the Hamming autocorrelation of $\bu$ is
$$H_a(\bu)=\max_{0<i<n}H_{\bu,\bu}(i)=\max_{0<i<n} \{n-d(\bu,\bu_i)\}=n-d(\bu).$$
We hope that $\bu_i$ for $0<i<n$ and $\bu$ are all distinct. Otherwise, we only have a trivial
bound $H_a(\bu)=n$. To serve this purpose, we come up with the concept of hopping cyclic codes.

\begin{defn}
A cyclic code $C\in \ZZ_q^n$ is called a hopping cyclic code if
for any $\bc\in C$, $\bc$ and any cycle shift of $\bc$ are all distinct.
\end{defn}

\begin{lemma}\label{lm:convert}
There is an $(n,M,\lambda;q)$ FH sequence set $\mF$ if and only if there is a $q$-ary $(n,nM,n-\lambda)$ hopping cyclic code.
\end{lemma}

\begin{proof}
Given any FH sequence set $\mF$ with parameters $(n,M,\lambda;q)$, we convert it into a hopping cyclic code $C$ in such a way that for each sequence $\bu\in \mF$, we add
$\bu$ together with all its cycle shift vectors $\bu_i,i=1,\ldots,n-1$ to $C$.
By the definition of $\mF$, the distance of this hopping cyclic code is $n-H_m(\mF)=n-\lambda$. Moreover, the size of $C$ is $n$ times the size of $\mF$. We can reverse this argument to prove the other direction.
\end{proof}

%

\section{Asymptotic Lower Bounds on FH sequences}
On the contrary to the upper bound, the lower bound on codes does not necessarily lead to the same lower bound on FH sequences.
Because we have to show the existence of hopping cyclic codes instead of any codes meeting this lower bound. In this sense, we need to prove the Gilbert-Varshamov bound on hopping cyclic codes. We provide two proofs, one based on probabilistic method that covers the whole rate region and another one is purely elementary that only covers part of the rate region.

\subsection{Probabilistic Method}

Our probabilistic method requires a tool called martingale.

\begin{defn}[\cite{Will}]
Let $\{X_i\}_{i=0}^{n-1}, \{Y_i\}_{i=0}^{n-1}$ be two sequences of  discrete random variables. We say that  $\{Y_i\}_{i=0}^{n-1}$ is a martingale with respect to $\{X_i\}_{i=0}^{n-1}$ if
\begin{itemize}
\item  $E(|Y_i|)<+\infty$, $0\leq i\leq n-1$.
\item  $E(Y_i|X_1,X_2,\cdots,X_{i-1})=Y_{i-1}$, $ 1\leq i\leq n-1$.
\end{itemize}
\end{defn}

\begin{lemma}[Azuma-Hoeffding inequality]
Suppose that $Y_{0},Y_{1},\ldots,Y_{n-1}$ is a martingale with respect to $X_0,\ldots,X_{n-1}$ such that
$|Y_{i+1}-Y_{i}|\leq c_{i}$.
Then, for any positive $t$, it holds that
\begin{equation*}
\Pr[Y_{n-1}-Y_{0}\geq t]\leq \exp\big(\frac{-t^2}{\sum_{i=0}^{n-2}2c_{i}^2}\big)
\end{equation*}
\end{lemma}

Recall that given $\bu=(u_0,\ldots,u_{n-1})\in \ZZ_q^n$, we use $\bu_i$ to represent the $i$-th cycle shift of $\bu$.
Our goal is to show that for most of the vectors in $\ZZ_q^n$, $\bu$ and $\bu_i,1\leq i\leq n-1$ are far away from each other.
As a warmup, we prove this claim for a prime number $n$.

\begin{lemma}\label{lm:du1}
Let $n$ be a prime number and $\Ge$ be any small constant.
If we pick $\mathbf{u}\in \ZZ_{q}^{n}$ uniformly at random, then we have
$$\Pr[d(\bu)\leq (n-2)(1-1/q-\epsilon)]\leq (n-1)e^{-\frac{\epsilon^2 (n-2)}{2}}.$$
\end{lemma}
\begin{proof}

We first fix $i$ and bound the probability over a random vector $\bu\in \ZZ_q^n$ that the minimum distance $d(\bu,\bu_i)\leq (n-1)(1-1/q-\epsilon)$.
Since $n$ is a prime number, we can write
$$
d(\bu,\bu_i)=|\{0\leq j\leq n-1: u_j\neq u_{j+i}\}|=|\{0\leq j\leq n-1: u_{ji}\neq u_{(j+1)i}\}|.
$$
Define random variable $X_{j}$ for $j=0,\ldots,n-1$ such that
\begin{enumerate}
\item $X_{j}=1$, if $u_{ji}=u_{(j+1)i}$\\
\item $X_{j}=0$,  otherwise.
\end{enumerate}
Let
\begin{eqnarray*}
Y_j=\sum\limits_{k=0}^{j}{X_k-\frac{j+1}{q}},  j=0,1,\ldots, n-2.
\end{eqnarray*}
Then, We claim that $\{Y_j\}_{j=0}^{n-2}$ is a martingale with respect to $\{X_j\}_{j=0}^{n-2}$.

First,it is clear that $E(|Y_j|)<+\infty$, for any $0\leq j\leq n-1$.

Next, as $n$ is a prime, for $0<j\leq n-2$, $u_0, u_i, u_{2i},\cdots,u_{(j-1)i}$ and $u_{ji}$ are independent random variables. This implies that
\begin{eqnarray*}
&&E(Y_{j}|X_{1},\ldots,X_{j-1})=E(X_j|X_0,\ldots,X_{j-1})+\sum_{i=0}^{j-1}X_i-(j+1)/q
\\
&&=\Pr[u_{(j+1)i}=u_{ji}|u_{0},u_i\ldots,u_{ji}]+\sum_{i=0}^{j-1}X_i-(j+1)/q=Y_{j-1}.
\end{eqnarray*}
Moreover, as $|Y_j-Y_{j-1}|\leq 1$ for $0\leq j\leq n-2$, the Azuma-Hoeffding inequality gives
$$
\Pr[Y_{n-2}-Y_0>\epsilon (n-2)]\leq \exp(-\frac{\epsilon^2 (n-2)}{2}).
$$
Equivalently, we have
$$
\Pr[\sum_{i=0}^{n-1}X_{j}>(\epsilon+\frac{1}{q})(n-2)+2]\leq \exp(-\frac{\epsilon^2 (n-2)}{2}).
$$
This implies $\Pr[d(\bu, \bu_i)\leq (n-2)(1-1/q-\epsilon)]\leq \exp(-\frac{\epsilon^2 (n-2)}{2})$.
Taking an union bound over $i=1,\ldots,n-1$, the desired result follows.
\end{proof}
This theorem can only be applied to the vectors of prime number length. Next, we will prove a weaker but asymptotically same\footnote{``Asymptotically same'' is referred to that it yields the same result when $n$ tends to infinity.} statement that holds for vectors of arbitrary length.
Before we state our theorem, we deviate a little to introduce the well-known Chernoff bound.
\begin{theorem}[Chernoff Bound]\label{thm:chern}
Let $X_1,\ldots,X_n$ be independent random variables such that $a\leq X_i\leq b$ for all $i$. Let $X=\sum_{i=1}^n X_i$ and set $\mu=E[X]$. Then, for $\delta>0$:
$$
\Pr[X>(1+\delta)\mu]\leq e^{-\frac{2\delta^2\mu^2}{n(b-a)^2}}, \quad \Pr[X<(1-\delta)\mu]\leq e^{-\frac{\delta^2\mu^2}{n(b-a)^2}}.
$$
\end{theorem}

\begin{lemma}\label{lm:nonprime}
Let $n$ be any positive integer and $\Ge$ be a small constant.
If we pick $\bu\in \ZZ_q^n$ uniformly at random, then we have
$$\Pr[d(\bu)\leq (n-2\sqrt{n})(1-1/q-\epsilon)]\leq n^2e^{-\frac{\epsilon^2 (\sqrt{n}-2)}{2}}.$$
\end{lemma}
\begin{proof}
As usual, we fix $0<i<n$ to be an integer and bound the probability that $\Pr[d(\bu,\bu_i)\leq (n-2\sqrt{n})(1-1/q-\epsilon)]$.
Let $d:=\gcd(i,n)$ and $r:=\frac{n}{d}$.
Partition $\{0,\ldots,n-1\}$ into $d$ disjoint sets, $S_j=\{j+ih: 0\leq h\leq r-1\}$ for $j=0,\ldots,d-1$.
It is clear that each $S_j$ has size $r$.
Let $X_h$ be the binary random variable such that $X_h=1$ if $u_{h}=u_{h+i}$ and $0$ otherwise.
It follows that $\Pr[X_h=1]=\frac{1}{q}$.
Let $Z_j=\sum_{h\in S_j}X_{h}$.
It is clear that $Z_0,Z_1,\ldots,Z_{d-1}$ are i.i.d random variables with
expectation
$$E[Z_j]=\sum_{h\in S_j}E[X_h]=\frac{|S_j|}{q}=\frac{r}{q}.$$
We divide our discussion into two cases:
\begin{itemize}
\item $d\leq \sqrt{n}$:  We focus on the binary random variable whose indices belong to the same subset $S_j$.
Define $Y_t=\sum_{h=0}^{t}X_{j+ih}-\frac{t+1}{q}$. The same argument in Lemma \ref{lm:du1} shows that
$Y_0,Y_1,\ldots,Y_{r-2}$ forms a martingale with respect to $X_j,X_{j+i},\ldots,X_{j+i(r-2)}$. That means,
$$
\Pr[Y_{r-2}-Y_0>\epsilon (r-2)]\leq \exp(-\frac{\epsilon^2 (r-2)}{2}).
$$
Or equivalently, we have
$$
\Pr[Z_j>(\epsilon+\frac{1}{q})(r-2)+2]\leq \exp(-\frac{\epsilon^2 (r-2)}{2}).
$$
Taking a union bound over $j=0,\ldots,d-1$, we have
$$
\Pr[\exists j, Z_j>(\epsilon+\frac{1}{q})(r-2)+2]\leq d\exp(-\frac{\epsilon^2 (r-2)}{2}).
$$
Notice that $d(\bu,\bu_i)=n-\sum_{t=0}^{d-1}Z_j$.
It follows that
$$
\Pr[d(\bu,\bu_i)\leq n-\big((\epsilon+\frac{1}{q})(r-2)+2\big)d]\leq d\exp(-\frac{\epsilon^2 (r-2)}{2})\leq n\exp(-\frac{\epsilon^2 (\sqrt{n}-2)}{2}).
$$
The desired result follows as
$$n-((\epsilon+\frac{1}{q})(r-2)+2\big)d=(n-2d)(1-\frac{1}{q}-\epsilon)\geq (n-2\sqrt{n})(1-\frac{1}{q}-\Ge).$$
\item $d>\sqrt{n}$: We have already shown that $Z_0,Z_1,\ldots,Z_{d-1}$ are independent random variables
such that $Z_i\in [0,r]$ and $E[Z_i]=\frac{r}{q}$. Let $Z=\sum_{i=0}^{d-1}Z_i$ and we obtain $E[Z]=\frac{n}{q}$. By Chernoff Bound, we have
$$
\Pr[Z\geq(1+q\Ge)\frac{n}{q}]\leq e^{-\frac{2n^2(q\Ge)^2}{q^2r^2d}}=e^{-2\Ge^2d}\leq e^{-2\Ge^2\sqrt{n}}.
$$
It follows that
$$
\Pr[d(\bu,\bu_i)\leq (1-\Ge-\frac{1}{q})n]\leq e^{-2\Ge^2\sqrt{n}},
$$
as $d(\bu,\bu_i)=n-Z$. This also implies that $\Pr[d(\bu,\bu_i)\leq (n-2\sqrt{n})(1-\frac{1}{q}-\Ge)]\leq e^{-2\Ge^2\sqrt{n}}$.
\end{itemize}
Taking a union bound over $i=1,\ldots,n-1$, we obtain the desired result.
\end{proof}

Lemma \ref{lm:nonprime} says that there exist at least $q^n(1-n^2e^{-\frac{\epsilon^2 (\sqrt{n}-2)}{2}})=q^n(1-o(1))$ vectors $\bu\in \ZZ_q^n$, each satisfying
$d(\bu)\geq (n-2\sqrt{n})(1-1/q-\epsilon)$.
Observe that the relative distance $(1-2/\sqrt{n})(1-1/q-\epsilon)$ tends to $(1-1/q-\epsilon)$ as $n$ grows.
It will not affect the asymptotic Gilbert-Varshamov bound.
We denote by $\mA$ the set of these vectors.
 Our next step is to obtain a class of hopping cyclic code approaching Gilbert-Varshamov bound. The proof is a standard argument of Gilbert-Varshamov bound except that we replace the
whole space $\ZZ_q^n$ with a relatively small set $\mA$ and we add the vector $\bc$ together with all its cycle shift vectors to the
code.

\begin{theorem}[Gilbert-Varshamov Bound on hopping cyclic code]\label{thm:GV}
Let $n$ be an integer and $\Ge$ be a small constant. There exists $(n,M,d<(n-2\sqrt{n})(1-1/q-\epsilon))$ hopping cyclic code
if
$$
M\leq \frac{q^n\big(1-q^n(1-n^2e^{-\frac{\epsilon^2 (\sqrt{n}-2)}{2}})\big)}{\sum_{i=1}^{d-1}\binom{n}{i}(q-1)^i}.
$$
Moreover, when $n\rightarrow \infty$, these hopping cyclic codes meeting the classic asymptotic Gilbert-Varshamov bound.
\end{theorem}
\begin{proof}.
Our task is to pick as many vectors as possible from set $\mA$ such that the distance between each other
is at least $d$. For $\bu\in \ZZ_q^n$, define
$V(\bu,d)=\{\bv\in \ZZ_q^n: d(\bv,\bu)<d\}$. The size of this Hamming ball is $V:=|V(\bu,d)|=\sum_{i=1}^{d-1}\binom{n}{i}(q-1)^i$.
Let $C\subseteq \mA$ be the hopping cyclic code with minimum distance $d$. We claim that
if $|C|V<|\mA|$, there still exists $\bu\in \mA-C$ such that $C\cup\{\bu,\bu_1,\ldots,\bu_{n-1}\}$ is a hopping cyclic code with minimum distance $d$. Since $|C|V<|\mA|$, there exists $\bu\in \mA$ such that $\bu$ is at least $d$ far away from any codewords in $C$ as
$$
|\mA-\bigcup_{\bc\in C}V(\bc,d)|\geq |\mA|-|C|V>0.
$$
Moreover, $C$ is already a hopping cyclic code. If $d(\bu_i,\bv)<d$ for some $\bv\in C$, then $d(\bu,\bv_{n-i})<d$ with $\bv_{n-i}\in C$. The contradiction occurs. This implies that $\bu_i$ is also at least $d$ far away from any codewords in $C$ for $1\leq i\leq n-1$.
Therefore, we can add $\bu$ to $C$ until $|C|V\geq|\mA|$. The desired result follows as
$|\mA|\geq q^n(1-n^2e^{-\frac{\epsilon^2 (\sqrt{n}-2)}{2}})$.
\end{proof}

The following corollary is a direct consequence of Lemma \ref{lm:convert} and Theorem \ref{thm:GV}.
\begin{cor}[Gilbert-Varshamov Bound I on FH sequences]\label{GV}
Let $n$ and $q>1$ be an integer. Let $\Ge$ be any small constant.
There exists FH sequences set $\mF$ with parameters $(n,M,\lambda;q)$ for $\lambda\geq n-(n-2\sqrt{n})(1-1/q-\epsilon)$
provided that
$$
nM\leq \frac{q^n\big(1-n^2e^{-\frac{\epsilon^2 (\sqrt{n}-2)}{2}}\big)}{\sum_{i=1}^{n-\lambda-1}\binom{n}{i}(q-1)^i}.
$$
When when $n\rightarrow \infty$, this implies the existence of FH sequences set with
information rate $r$ and relative Hamming correlation $\delta_H>\frac{1}{q}+\Ge$ satisfying
\begin{eqnarray*}
r\geq 1-H_q(1-\delta_H).
\end{eqnarray*}
\end{cor}
%
%

\subsection{Elementary Method}

In this subsection, we present another proof of Gilbert-Varshamov bound which does not rely on any advanced tool. Although this proof only covers part of rate region in Gilbert-Varshamov bound, we believe that it helps to understand this bound from another angle.
Our first step is to bound the size of set $\{\bu\in F^n: d(\bu)<d\}$ for a given alphabet set $F$.

\begin{lemma}
Let $F$ be a finite alphabet of size $q>1$. For $d<(1-\frac{1}{\sqrt{q}+1})n$,
\begin{equation*}
|\{\bu\in F^n: d(\bu)<d\}|\leq  nd\binom{n}{d-1}q^{\frac{n+d-1}{2}}.
\end{equation*}
\end{lemma}
\begin{proof}
We fix $0<i<n$ and bound the size of set $\mA_i:=\{\bu\in F^n: d(\bu,\bu_i)<d\}$.
Let $\bu=(u_0,\ldots,u_{n-1})\in \mA_i$.
Since $d(\bu,\bu_i)<d$, there exists an index subset of size $k\geq n-d+1$ such that $\bu$ and $\bu_i$ agree on these $k$ indices.
There are $\binom{n}{k}$ distinct subsets of size $k$. We fix $T=\{j_1,\ldots,j_{k}\}$ to be one of them and count the number of $\bu$ that satisfies this requirement. Note that $\bu$ is subject to $u_{j_\ell}=u_{j_\ell+i}$ for $1\leq \ell\leq k$.
For each $\ell$, let $r$ be the largest number such that $\{j_\ell, j_\ell+i,\ldots,j_\ell+i(r-1)\}\subseteq T$.
This implies that
$$u_{j_\ell}=u_{j_\ell+i}=u_{j_\ell+2i}=\cdots=u_{j_\ell+ri}.$$
Define $[j_\ell]:=\{j_\ell, j_\ell+i,\ldots,j_\ell+i(r-1),j_\ell+ir\}$. It is clear that if $j_r\in [j_\ell]$, then $[j_r]\subseteq [j_\ell]$.
Let $[\ell_1],\ldots, [\ell_a]$ be the maximal subsets from $[j_1],\ldots,[j_k]$. Observe that they are disjoint and each set has size at least $2$.
Define $S=\bigcup_{i=1}^a[\ell_i]$. It is clear that $T\subseteq S$ and $a\leq \frac{|S|}{2}$. Observe that if $j\in S$, then
$u_j=u_{\ell_t}$ for some $1\leq t\leq a$.
This implies if we fix $u_{\ell_1},\ldots,u_{\ell_a}$ and $u_i$ for $i\in [n]\backslash S$, $\bu$ is completely determined.
Thus, there are $a+n-|S|\leq n-\frac{|S|}{2}\leq n-\frac{k}{2}$ indices to be fixed which yields at most $q^{n-k/2}$ distinct $\bu$.
Let $k$ run from $n-d+1$ to $n$, the number of $\bu$ is upper bounded by
$
\sum_{k=n-d+1}^n \binom{n}{k}q^{n-\frac{k}{2}}.
$

We proceed to bound this value.
Observe that
$$
\frac{\binom{n}{k}q^{n-k/2}}{\binom{n}{k-1}q^{n-(k-1)/2}}=(\frac{n-k+1}{k})q^{-\frac{1}{2}}.
$$
This means if $k\geq \frac{n+1}{\sqrt{q}+1}$, function $\binom{n}{k}q^{n-\frac{k}{2}}$ is monotone decreasing with respect to $k$.
Since $k\geq n-d+1> \frac{n+1}{\sqrt{q}+1} $,
this function reaches its maximum when $k=n-d+1$.
It follows that
$$
\sum_{k=n-d+1}^n \binom{n}{k}q^{n-\frac{k}{2}}\leq d\binom{n}{d-1}q^{\frac{n+d-1}{2}}.
$$
This implies $|\mA_i|\leq d\binom{n}{d-1}q^{\frac{n+d-1}{2}}$.
Observe that
\begin{equation*}
\{\bu\in F^n: d(\bu)<d\}=\bigcup_{0<i<n}\mA_i.
\end{equation*}
The desired result follows.
\end{proof}

To apply the argument of Gilbert-Varshamov bound in Theorem \ref{thm:GV}, we need to show that
$nd\binom{n}{d-1}q^{\frac{n+d-1}{2}}$ is at most $o(q^n)$.
\begin{lemma}
Assume that $d<(1-\frac{e}{\sqrt{q}})n$, then $nd\binom{n}{d-1}q^{\frac{n+d-1}{2}}=o(q^n)$.
\end{lemma}
\begin{proof}
Observe that
$$
\frac{nd\binom{n}{d-1}q^{\frac{n+d-1}{2}}}{q^n}\leq nd(2^{H_2(\frac{d-1}{n})}q^{\frac{d-n-1}{2n}})^n=2^{n\big(H_2(\frac{d}{n})-\frac{n-d}{2n}\log_2 q+o(1)\big)}.
$$
Let $\delta=\frac{d}{n}$ and then $\delta\leq 1-\frac{e}{\sqrt{q}}$.
It suffices to show that $H_2(\delta)<\frac{1-\delta}{2}\log_2 q$.
Note that
\begin{eqnarray*}
H_2(\delta)&=&-\delta\log_2 \delta-(1-\delta)\log_2 (1-\delta)\\
&\leq&-\delta\log_2 \delta-(1-\delta)\log_2 (\frac{e}{\sqrt{q}})=-\delta\log_2 \delta-(1-\delta)\log_2 e+\frac{1-\delta}{2}\log_2 q.
\end{eqnarray*}
The first inequality is due to $\delta\leq 1-\frac{e}{\sqrt{q}}$.
Let $f(\delta)=-\delta\log_2 \delta-(1-\delta)\log_2 e$. The derivative of $f(\delta)$ is
$-\log_2 \delta-\delta\log_2 e+\log_2 e$ which is always positive when $\delta\in (0,1)$.
That means $f(\delta)<f(1)=0$ for $\delta\leq 1-\frac{e}{\sqrt{q}}$. The desired result follows.
\end{proof}

The same argument in Theorem \ref{thm:GV} gives the Gilbert-Varshamov bound on FH sequences set. Although it will not affect the asymptotic behaviour of Gilbert-Varshamov, the size of FH sequences set is bigger than the size in Corollary \ref{GV}.

\begin{cor}[Gilbert-Varshamov Bound II on FH sequences]
Let $n$ and $q>1$ be an integer.
There exists FH sequences set $\mF$ with parameters $(n,M,\lambda\geq \frac{en}{\sqrt{q}};q)$
provided that
$$
nM\leq \frac{q^n-(n-\lambda)\binom{n}{n-\lambda-1}q^{\frac{2n-\lambda-1}{2}}}{\sum_{i=1}^{n-\lambda-1}\binom{n}{i}(q-1)^i}.
$$
When $n\rightarrow \infty$, this implies the existence of FH sequences set with
information rate $r$ and relative Hamming correlation $\delta_H\geq \frac{e}{\sqrt{q}}$ satisfying
\begin{eqnarray*}
r\geq 1-H_q(1-\delta_H).
\end{eqnarray*}
\end{cor}

\section{Conclusion}
In this paper, we present the asymptotic upper bounds and lower bound on the FH sequences. Our upper bounds are derived from \cite{Ding09} by letting sequence length tend to infinity. Our lower bound comes from the classic Gilbert-Varshamov bound. However, the original Gilbert-Varshamov bound does not hold for cyclic codes. The main contribution of this paper is dedicated to showing the existence of hopping cyclic codes attaining this bound. Since a hopping cyclic code can be converted into a set of FH sequence, this result also implies the existence of FH sequences attaining this bound. As we know, there exists a class of algebraic geometry codes that can beat Gilbert-Varshamov bound over a field of size $q\geq 49$ \cite{AG}. Unfortunately, this class of algebraic geometry codes is not cyclic. It seems less likely to convert it into FH sequences.
There then arises an open problem whether there exists a class of FH sequences beating Gilbert-Varshamov bound.





\begin{thebibliography}{99}

\bibitem{Fan}P. Fan and M. Darnell, \textit{Sequence Design for Communications Applications.
London,} U.K.: Wiley, 1996.

\bibitem{Golomb} S. W. Golomb, G. Gong,  \textit{Signal design for good correlation: for wireless communication, cryptography and radar.} U.K.: Cambridage Univ. Press, 2005.

\bibitem{Simon}M. K. Simon, J. K. Omura, R. A. Scholtz, and B. K. Levitt, \textit{Spread Spectrum
Communications Handbook.} New York, NY, USA: McGraw-Hill, 1994.

\bibitem{LG} A. Lempel, H. Greenberger, Families sequence with optimal
Hamming correlation properties, \textit{IEEE Trans. Inf. Theory}, vol.IT-20,
pp.90-94, Jan. 1974.



\bibitem{PF} D. Y. Peng and P. Fan, Lower bounds on the Hamming auto- and
cross correlations of frequency-hopping sequences, \textit{IEEE Trans. Inf.
Theory}, vol. 50, no. 9, pp. 2149-2154, 2004.

\bibitem{Ding09} C. Ding, R. Fuji-Hara, Y. Fujiwara, M. Jimbo, and M. Mishima, Sets
of frequency hopping sequences: Bounds and optimal constructions, \textit{IEEE Trans. Inf.
Theory}, vol. 55, no. 7, pp. 3297-3304, 2009.

\bibitem{Yang} Y. Yang, X.H. Tang, P. Udaya, and D.Y. Peng, New bound on frequency
hopping sequence sets and its optimal constructions, \textit{IEEE Trans. Inf.
Theory}, vol. 57, no. 11, pp. 7605-7613, 2011.

\bibitem{Liu} X. Liu, D. Peng, H. Han, Improved Singleton Bound on Frequency Hopping Sequences, Springer Sequences and Their Applications-SETA,  2014, pp. 305-314.

\bibitem{Xing} S. Lin and C. Xing, \textit{Coding Theory: A First Course}, Cambridge, U.K.: Cambridge Univ. Press, 2004.

\bibitem{AG} M.A. Tsfasman and S.G. Vladut. \textit{Algebraic-Geometric codes}. Kluwer
Academic Publisher, 1991.


\bibitem{Will} D. Williams, \textit{Probability with Martingale}, Cambridge university press, 1991.

\bibitem{Ding10} C. Ding, Y. Yang, and X. H. Tang, Optimal sets of frequency hopping
sequences from linear cyclic codes, \textit{IEEE Trans. Inf.
Theory}, vol. 55, no. 7, pp. 3605-3612, 2010.
\end{thebibliography}
\end{document}